\documentclass[a4paper]{article}
\pdfoutput=1

\usepackage{amsmath, amsthm, amssymb}
\usepackage{graphicx,color} 
\usepackage{authblk}

\theoremstyle{plain} 
\newtheorem{theorem}{Theorem}

\newtheorem{fact}{Fact}

\newtheorem{proposition}{Proposition}

\begin{document}

\title{Drawing a Rooted Tree as a Rooted $y-$Monotone Minimum Spanning Tree}

\author{Konstantinos Mastakas}

\affil{School of Applied Mathematical and Physical Sciences\\ National Technical University of Athens, Athens, Greece \\ kmast@math.ntua.gr}

\date{}

\maketitle

\begin{abstract}
Given a rooted point set $P$, the \emph{rooted $y-$Monotone Minimum Spanning Tree} (rooted $y-$MMST) of $P$ is the spanning geometric graph of $P$ in which all the vertices are connected to the root by some $y-$monotone path and the sum of the Euclidean lengths of its edges is the minimum.
We show that the maximum degree of a rooted $y-$MMST is not bounded by a constant number. We give a linear time algorithm that draws any rooted tree as a rooted $y-$MMST and also show that there exist rooted trees that can be drawn as rooted $y-$MMSTs only in a grid of exponential area.
\end{abstract}

\section{Introduction}

Given a point set $P$ with $n$ points, the \emph{Euclidean Minimum Spanning Tree (EMST)} of $P$, i.e.{} the spanning geometric graph of $P$ where the sum of the Euclidean lengths of its edges is the minimum, can be obtained in $O(n\cdot\log n)$ time~\cite{ShaH75}. 
A Euclidean Minimum Spanning Tree of a point set is of maximum degree at most six~\cite{MonS92} and a Euclidean Minimum Spanning Tree of maximum degree at most five can always be found~\cite{MonS92}.
The problem of drawing a tree of maximum degree at most five as a EMST can be efficiently solved~\cite{MonS92}.
On the other hand, for a tree of maximum degree six the same problem is NP-hard~\cite{EadW96}.
Regarding the area requirement of the drawing of a tree of maximum degree five as a EMST, the algorithm of Monma and Suri~\cite{MonS92} produces a drawing in a grid of exponential area. 
Additionally, there exist trees of maximum degree at most five for which no drawing as EMST lies on a grid of polynomial area~\cite{AngBCFKS14}.

Given a rooted point set $P$ of size $n$ with root $r$, the spanning geometric graph of $P$ in which each point of $P$ is connected with $r$ by a $y-$monotone path and the sum of the Euclidean lengths of its edges is the minimum is a tree~\cite{MasS17}, denoted as the \emph{rooted $y-$Monotone Minimum Spanning Tree (rooted $y-$MMST)} of $P$ in~\cite{MasS17}, and can be obtained in $O(n\log^2 n)$ time~\cite{MasS17}.
In contrast to the case of Euclidean Minimum Spanning Tree, it is not known whether the maximum degree of a rooted $y-$MMST is bounded by a constant number,  whether a drawing of a rooted tree as a rooted $y-$MMST can be efficiently produced and if such a drawing necessarily lies on a grid of exponential area.

The restricted fathers tree problem which was studied by Guttman-Beck and Hassin~\cite{Gut-BecH10} is related to the problem of obtaining the rooted $y-$MMST of a rooted point set. 
The goal of the \emph{restricted fathers tree problem} is to obtain the minimum spanning tree $T$ of a weighted rooted graph $G$ where each vertex of $G$ contains a key, where in $T$ the root is connected to all other vertices by paths in which the keys of the traversed vertices form a decreasing sequence.
The restricted fathers tree problem is greedily solvable~\cite[Corollary 2.6]{Gut-BecH10}.

Much research has been done in drawing a rooted tree under several aesthetical drawing conventions~\cite[Section 3.1]{DiBatETT99}.
The problem of finding a drawing w.r.t.{} minimizing the area is widely studied~\cite{DiBatF14}.
Recently, Chan~\cite{Cha18} improved the area requirements of several types of tree drawings.
As far as monotonicity is concerned, drawing a rooted tree such that each child vertex is mapped to a point with $y$ coordinate less than or equal to the $y$ coordinate of the point to which its parent is mapped with the goal of optimizing the area of the drawing is widely investigated, e.g.{} see~\cite[Chapter 3]{Shi76},~\cite{CreDiBP92},~\cite{Cha18}.
Furthermore, drawing a rooted tree in the plane such that each pair of points is connected by a path that is monotone in some direction, in a grid of small area, is thoroughly investigated~\cite{AngCBFP12,HeH17,OikS17}.

In this article, we show that the maximum degree of a rooted $y-$MMST is not bounded by any constant number. We give a linear time recursive algorithm that draws a rooted tree as a rooted $y-$MMST and we show that there exist rooted trees that can be drawn as a rooted $y-$MMST only in a grid of exponential area (and not in a grid of polynomial area).

\section{Preliminaries}

Let $a,b$ points of the plane. The vector from $a$ to $b$ is denoted as $\overrightarrow{ab}$.

A \emph{geometric graph} $G = (P,L)$ is a pair of (i) a point set $P$ which is its \emph{vertex} set and (ii) a set of line segments $L$ connecting points of $P$ which is its \emph{edge} set.
A \emph{geometric path} $(p_1$, $p_2$, \ldots, $p_k)$ is a geometric graph with $\{p_1$, $p_2$, \ldots, $p_k\}$ as its vertex set and $\{\overline{p_1p_2}$, $\overline{p_2p_3}$, \ldots, $\overline{p_{k-1}p_k}\}$ as its edge set.
The geometric path $(p_1$, $p_2$, \ldots, $p_k)$ is $y-$monotone if either (i) for each $i =1$, $2$, \ldots, $k-1$ the $y$ coordinate of $p_{i+1}$ is greater than or equal to the $y$ coordinate of $p_i$, or (ii) for each $i =1$, $2$, \ldots, $k-1$ the $y$ coordinate of $p_{i+1}$ is less than or equal to the $y$ coordinate of $p_i$.

Lee and Preparata~\cite{LeeP77} showed that in a geometric graph where all vertices have different $y$ coordinates, each vertex $u$ is connected to the vertex $B$ with the smallest $y$ coordinate by some $y-$monotone path if and only if each vertex $u \neq B$ is connected to some vertex $v$ where the $y$ coordinate of $v$ is smaller than the $y$ coordinate of $u$.

\begin{fact}[Corollary 3 of~\cite{MasS17}\footnote{It was recently brought to our attention that Fact~\ref{fact:yMMST-char} can also be obtained as a Corollary of~\cite[Corollary 2.6]{Gut-BecH10}.}] \label{fact:yMMST-char}
Given a geometric tree $T$ with root $r$ where all vertices have different $y$ coordinates and each vertex different from the root has $y$ coordinate greater than the $y$ coordinate of the root, $T$ is the rooted $y-$MMST of its vertex set if and only if each vertex different from the root is linked to a single vertex below it, its nearest vertex.  
\end{fact}

The \emph{depth} of a rooted tree $T$ is the maximum number of edges traversed in a path from the root of $T$ to another vertex of $T$.

A \emph{straight-line drawing} $D$ of a graph $G = (V,E)$ is a function that associates each $u \in V$ to a point of the plane and each $\{u,v\} \in E$ to a line segment of the plane connecting the images of $u$ and $v$.
In this article, we only deal with straight-line drawings, hence we omit the term straight-line from now on.  

\section{Drawing a Rooted Tree as a Rooted $y-$Monotone Minimum Spanning Tree}

In this section we show that for each natural number $M$ we can find a rooted point set $P$ for which the rooted $y-$MMST of $P$ is of maximum degree $M$. 
We also give a linear time recursive algorithm that draws a tree $T$ with root $r$ as a rooted $y-$MMST. 
Finally, we prove that there exist rooted trees that can be drawn as rooted $y-$MMSTs only in a grid of exponential area.

\begin{proposition} \label{prop:UnboundedDegree}
The maximum degree of a rooted $y-$MMST is not bounded by a constant number.
\end{proposition}

\begin{proof}
For each $M \in \mathbb{N}$, we consider a rooted point set with $M+1$ points, in which all points different from the root $r$ are higher than $r$ and in increasing order w.r.t.{} the $x$ coordinate and decreasing order w.r.t.{} the $y$ coordinate. 
More specifically, the point set consists of the points $r = (0,0)$, $p_1 = (x_1 = 1, M)$, $p_2  = (x_2, M-1)$, \ldots, $p_i = (x_i, M+1-i)$, \ldots, $p_M = (x_M, 1)$ where $x_{i+1} = x_i + 1 + \lfloor\sqrt{x^2_i+(M+1-i)^2 - 1}\rfloor$, $i = 1$, $2$, \ldots, $M-1$.
Then, for each $p_i$, $i = 1$, $2$, \ldots, $M$, the closest point of $P$ that has $y$ coordinate less than the $y$ coordinate of $p_i$ is the root, since for each $1 \leq i \leq M-1$, $d(p_i, p_{i+k}) \geq d(p_i, p_{i+1})$ $= \sqrt{(x_{i+1} - x_i)^2+1}$ $ > \sqrt{x^2_i+(M+1-i)^2} = d(p_i, r)$ and for $p_M$ the root is the single point of $P$ with less $y$ coordinate than it.
Hence, by Fact~\ref{fact:yMMST-char} the rooted $y-$MMST of $P$ is the tree where each $p_i$, $1 \leq i \leq M$, is connected just to $r$.
\end{proof}

We now give our algorithm that draws a rooted tree $T$ as a rooted $y-$MMST.
Our algorithm is recursive.
It first draws the subtrees $T_1$, $T_2$, \ldots, $T_M$ with roots $p_1$, $p_2$, \ldots, $p_M$, that are connected to the root $r$, as rooted $y-$MMSTs and then computes appropriate vectors $\overrightarrow{rp_1}$, $\overrightarrow{rp_2}$, \ldots, $\overrightarrow{rp_M}$ such that the final drawing is a rooted $y-$MMST. 
In each vertex $u \neq r$ of $T$ we assign the vector $\overrightarrow{p(u)u}$, where $p(u)$ denotes the parent of $u$ in $T$.
We also store in each vertex $u$ of $T$ the width (\emph{width$(u)$}) and the height (\emph{height$(u)$}) of the \emph{bounding box}, i.e.{} the box of the minimum area in which the drawing is included, of the subtree rooted at $u$. 

If $T$ has depth $1$, and the children of $r$ are the vertices $p_1$, $p_2$, \ldots, $p_M$, then we use the same construction that we used in Proposition~\ref{prop:UnboundedDegree}. 
In more detail, we assign $\overrightarrow{rp_1} = (x_1 = 1,M)$, $\overrightarrow{rp_2} = (x_2, M-1)$, \ldots, $\overrightarrow{rp_i} = (x_i, M+1-i)$, \ldots, $\overrightarrow{rp_M} = (x_M, 1)$ where $x_{i+1} = x_i + 1 + \lfloor\sqrt{x^2_i+(M+1-i)^2 - 1}\rfloor$, $1 \leq i \leq M-1$.
Furthermore, height$(r) = M$ and width$(r) = x_M$.
For an example see Figure~\ref{fig:illDepth1}.
In this way, as shown in Proposition~\ref{prop:UnboundedDegree}, the root $r$ is the closest point below for each $p_i$, $i = 1$, $2$, \ldots, $M$, and hence by Fact~\ref{fact:yMMST-char} our drawing is a rooted $y-$MMST. 

\begin{figure}[hbtp]
\center
\includegraphics{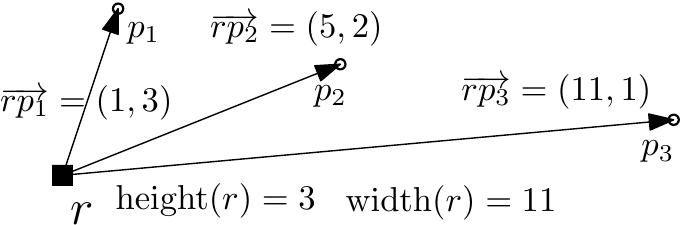}
\caption{Illustration of the drawing of a rooted tree of depth $1$ where the root has $3$ children.}
\label{fig:illDepth1}
\end{figure}

If the tree has depth at least $2$, let $p_1$, $p_2$, \ldots, $p_M$ be the children of $r$ and $T_1$, $T_2$, \ldots, $T_M$ be the corresponding subtrees.
We recursively draw the subtrees $T_1$, $T_2$, \ldots, $T_M$ and then we assign the values of the vectors $\overrightarrow{rp_1}$, $\overrightarrow{rp_2}$, \ldots, $\overrightarrow{rp_M}$ in analogous way to the case that the tree has depth $1$, but this time the $x$ and $y$ coordinates of the vectors are larger such that the final drawing is a rooted $y-$MMST.
More specifically, let $\overrightarrow{rp_i} = (x_i, y_i)$, $i$ $=$ $1$, $2$,\ldots, $M$ then $y_M = 1$, $y_{i-1}$ $=$ $y_i$ $+$ height$(p_i) + 1$, $i$ $=$ $M$, $M-1$, \ldots, $2$, since each $T_{i-1}$ is drawn above $T_i$.
Regarding the $x$ coordinates of the vectors, in order that the closest point below for the vertices in the final drawing is their parent and hence by Fact~\ref{fact:yMMST-char} the final drawing is a rooted $y-$MMST, the $x$ coordinates of the vectors satisfy the following recursive formula:
 $x_1 = 1$ and $x_{i+1} = \max\{x_i$ $+$ $1$ $+$ $\lfloor\sqrt{x^2_i + y^2_i - 1}\rfloor$, $x_i$ $+$ width$(p_i)$ $+$ $1$ $+$ $\lfloor\sqrt{\textrm{width}^2(p_i)+\textrm{height}^2(p_i)-1}\rfloor \}$.
The $x$ coordinates of the vectors satisfy this recursive formula because 
(i) $x_{i+1}$ has to be greater than or equal to $x_i$ $+$ $1$ $+$ $\lfloor\sqrt{x^2_i + y^2_i - 1}\rfloor$, $i$ $=$ $1$, $2$,\ldots, $M-1$, such that for each vertex $u$ in $T_{i+1}$ it holds that $d(p_i,u)$ $\geq$ $\sqrt{(x_{i+1}-x_i)^2+1}$ $>$ $d(p_i,r)$ and hence the closest point to $p_i$ below $p_i$ is $r$, and (ii) $x_{i+1}$ has to be greater than or equal to $x_i$ $+$ width$(p_i)$ $+$ $1$ + $\lfloor\sqrt{\textrm{width}^2(p_i)+\textrm{height}^2(p_i)-1}\rfloor$, 
$i$ $=$ $1$, $2$,\ldots, $M-1$, such that for each vertex $u_i$ in $T_{i}\setminus\{p_i\}$ and each $u_{i+1}$ in $T_{i+1}$ it holds that $d(u_i,u_{i+1})$ $\geq$ $\sqrt{(x_{i+1}-x_i-\textrm{width}(p_i))^2+1}$ $>$ $d(u_i,p_i)$ and hence the closest point below to each vertex $u_i$ in $T_{i}\setminus\{p_i\}$ remains its parent which lies inside the drawing of $T_i$.
For an example see Figure~\ref{fig:illOfScaling}.   
From the above discussion, we get the following Theorem.
\begin{figure}[hbtp]
\center
\includegraphics[scale=0.55]{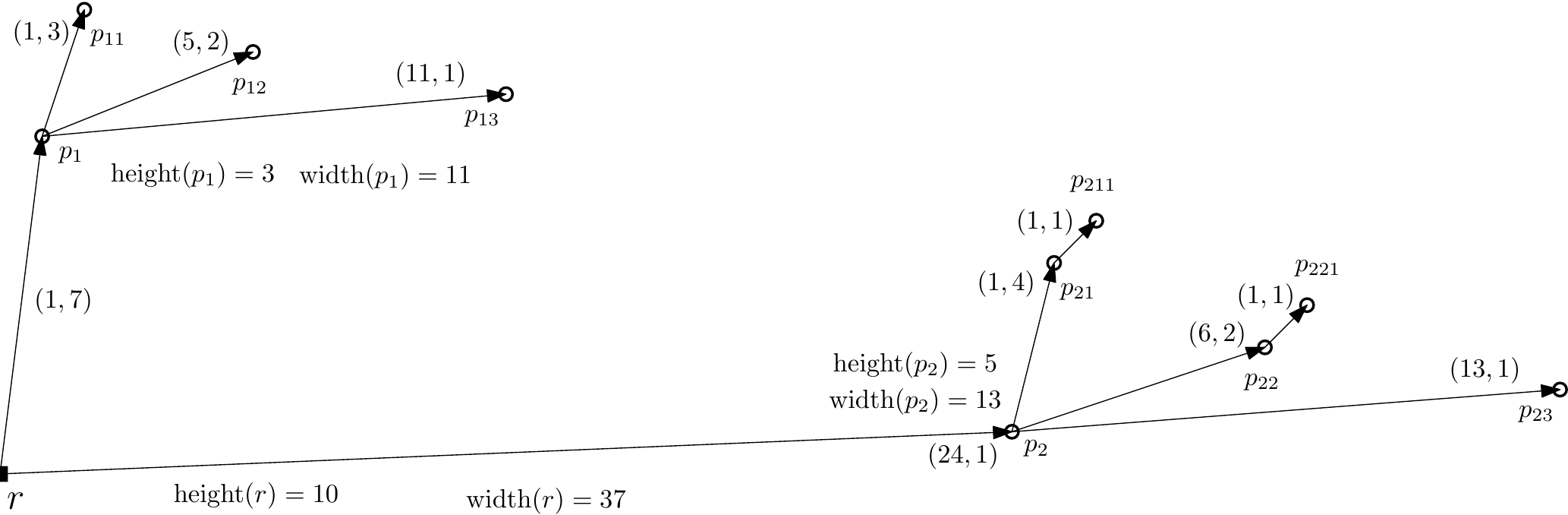}
\caption{
When drawing the subtree rooted at $p_2$, for the vector $\protect\overrightarrow{p_2p_{22}} = (x_2,y_2)$ it holds that $x_2$ $=$ $x_1$ $+$ $1$ $+$ $\lfloor\sqrt{x^2_1 + y^2_1 - 1}\rfloor$ $=$ $1$ $+$ $1$ $+$ $\lfloor\sqrt{1^2 + 4^2 - 1}\rfloor$ $=$ $6$.
In the final drawing, for the vector $\protect\overrightarrow{rp_2} = (x'_2,y'_2 = 1)$ it holds that $x'_2$ = $1$ $+$ width$(p_1)$ $+$ $1$ $+$ $\lfloor\sqrt{\textrm{width}^2(p_1)+\textrm{height}^2(p_1)-1}\rfloor $ $=$ $1$ $+$ $11$ $+$ $1$ $+$ $\lfloor\sqrt{11^2+3^2-1}\rfloor$ $=$ $24$.
}
\label{fig:illOfScaling}
\end{figure}

\begin{theorem}\label{thm:DrawingResult}
We can draw a rooted tree as a rooted $y-$Monotone Minimum Spanning Tree in linear time.
\end{theorem}

We now show that drawing some rooted trees as rooted $y-$MMSTs can not be performed in a grid of polynomial area. 

\begin{theorem} \label{thm:exponentialGrid} There exist rooted trees which can be drawn as rooted $y-$Monotone Minimum Spanning Trees only in a grid of exponential area.
\end{theorem}

\begin{proof}
We show that drawing the rooted tree $K_{1,4n+4}, n \in \mathbb{N}$, i.e.{} the tree consisting of a root center vertex $r$ and $4n+4$ leaf vertices which are adjacent only to $r$, as a rooted $y-$MMST requires a grid of exponential width.
Consider, w.l.o.g.{} that $r$ is drawn at the origin of the Cartesian Coordinate System.
Then, by the pigeonhole principle at least $n+1$ of the vertices are all drawn in the same quadrant of the Cartesian Coordinate System.
Assume w.l.o.g.{} that this quadrant is the first quadrant, i.e.{} the quadrant containing the points with positive $x$ and $y$ coordinates.
By Fact~\ref{fact:yMMST-char} it follows that no two vertices can be mapped to points of the first quadrant with the same $x$ coordinate.
 Let the vertices, mapped in the first quadrant, be mapped to the points $p_1$, $p_2$, \ldots, $p_{n+1}$, with $x(p_i) < x(p_{i+1})$, $i$ $=$ $1$, $2$, \ldots, $n$.
By Fact~\ref{fact:yMMST-char} and since all the points $p_1$, $p_2$, \ldots, $p_{n+1}$ belong to the first quadrant, it follows that $y(p_1) > y(p_2) > \dots > y(p_{n+1})$ and hence $y(p_i) \geq n+1-i$.
We now show that the drawing $D$ restricted to the first quadrant has exponential area.
We do this as follows: We first show that $D$ has more area than another drawing $D'$, where $D'$ is a rooted $y-$MMST and that $D'$ has exponential width.
The intuition is to  ``translate'' the points $p_i$, $i$ $=$ $1$, $2$, \ldots, $n+1$, only in the $y$ coordinate as low as possible in order to obtain $D'$.
More specifically, $D'$ consists of the root mapped at the origin, and $p'_i$ where $x(p'_i) = x(p_i)$ and $y(p'_i) = n+1-i$, $i = 1,2, \ldots, n+1$ and the line segments $\overline{rp'_i}$, $i = 1,2, \ldots, n+1$. 
See for example Figure~\ref{fig:illustrationOfExponentialGrid}(a).
Clearly, $D'$ has less area than $D$.
Furthermore, since the circle $C_i$ with center $p_i$ and radius $d(p_i,r)$ contains in its interior the part of the circle $C'_i$ with center $p'_i$ and radius $d(p'_i,r)$ that is inside the first quadrant and both $p_{i+1}$ and $p'_{i+1}$ are outside $C_i$, then $d(p'_{i+1}, p'_i) > d(p'_i,r)$ (see for example Figure~\ref{fig:illustrationOfExponentialGrid}(b)). Hence, by Fact~\ref{fact:yMMST-char} it follows that $D'$ is a rooted $y-$MMST. 
In the $D'$, it is easy to show that the width is exponential. Since $(x(p'_{i+1})-x(p'_i))^2+1 \geq x(p'_i)^2+(n+1-i)^2$ it follows that $x(p'_{i+1}) \geq 2 x(p'_i), i = 1,2, \ldots, n$, and $x(p'_2) \geq n-1$.  
\end{proof}

\begin{figure}[hbtp]
\center
\begin{tabular}{cc}
\includegraphics[scale=0.7]{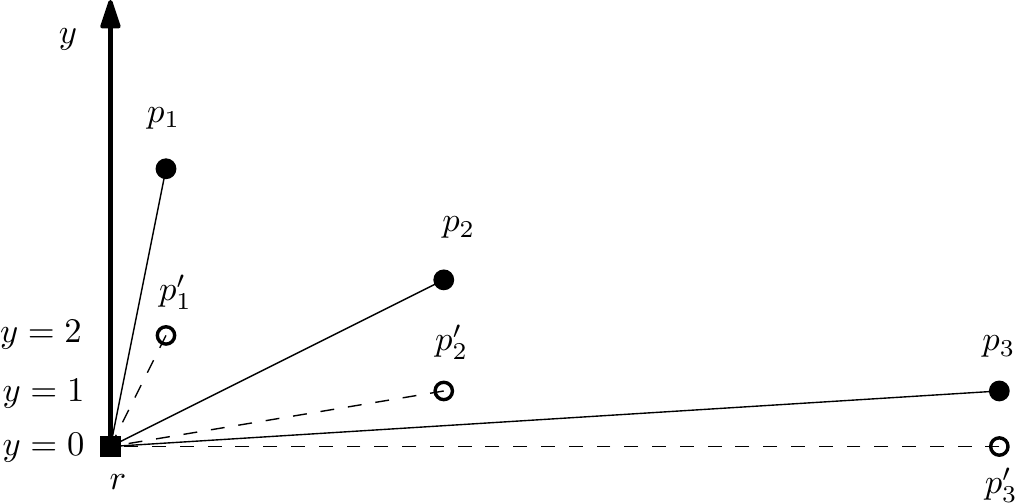} & \includegraphics[scale=0.5]{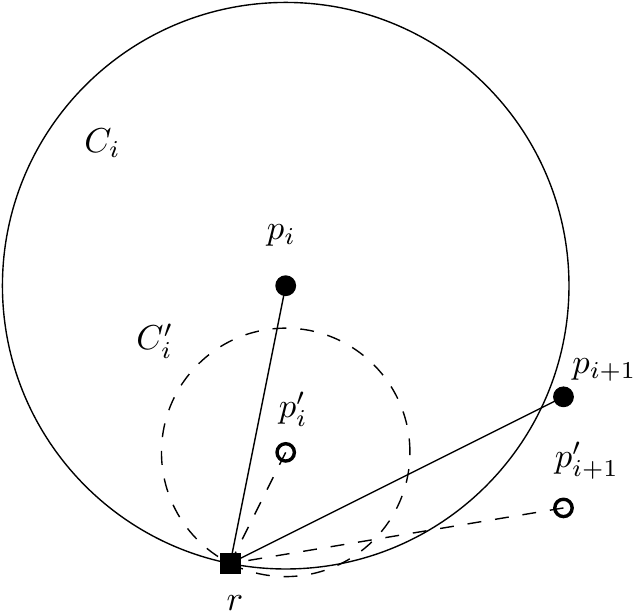} \\
(a) & (b)
\end{tabular}
\caption{Illustration of Theorem~\ref{thm:exponentialGrid}. In (a) the new drawing is illustrated and in (b) the feasibility of the new point $p'_{i+1}$ is shown based on the circles $C_i$ and $C'_i$.}
\label{fig:illustrationOfExponentialGrid}
\end{figure}

\section{Conclusions}

In this article, we showed that the maximum degree of a rooted $y-$MMST is not bounded by any constant. We gave a linear time recursive algorithm that draws a rooted tree as a rooted $y-$MMST in a grid of exponential area. We finally proved that there exist rooted trees that can be drawn as rooted $y-$MMSTs only in a grid of exponential area. 

\noindent \textbf{Acknowledgments.} The author would like to thank the Professor Alexander Arvanitakis and the Professor Aris Pagourtzis for the thoughtful discussions.
This research was financially supported by the Special Account for Research Grants of the National Technical University of Athens. 

\bibliographystyle{plain}
\bibliography{lib}

\end{document}